\documentclass[11pt,a4paper]{article}

\usepackage[T1]{fontenc}
\usepackage[utf8]{inputenc}
\usepackage[UKenglish]{babel}
\usepackage{lmodern}
\usepackage{microtype}
\usepackage[a4paper,margin=1in]{geometry}
\usepackage{amsmath,amssymb,amsthm}
\usepackage{thmtools}
\usepackage{enumerate}
\usepackage{graphicx}
\usepackage{xcolor}
\usepackage{tikz}
\usepackage{authblk}
\usepackage[hidelinks]{hyperref}

\hypersetup{
  pdftitle={Geometric Optimization Parameterized by Piercing Complexity},
  pdfauthor={Aritra Banik, Rajiv Raman, and Saurabh Ray}
}

\theoremstyle{plain}
\newtheorem{theorem}{Theorem}
\newtheorem{lemma}[theorem]{Lemma}
\newtheorem{corollary}[theorem]{Corollary}
\theoremstyle{definition}
\newtheorem{definition}[theorem]{Definition}
\theoremstyle{remark}
\newtheorem{remark}[theorem]{Remark}
\theoremstyle{plain}
\newtheorem*{theorem*}{Theorem}

\newcommand{\poly}{\text{poly}}
\newcommand{\arm}{\text{arm}}

\title{Geometric Optimization Parameterized by Piercing Complexity}
\author[1,2]{Aritra Banik\thanks{\texttt{aritra@niser.ac.in};
\href{https://orcid.org/0000-0002-7544-6125}{ORCID 0000-0002-7544-6125}}}
\author[3]{Rajiv Raman\thanks{\texttt{rajiv@iiitd.ac.in};
\href{https://orcid.org/0009-0000-8013-9421}{ORCID 0009-0000-8013-9421}}}
\author[4]{Saurabh Ray\thanks{\texttt{saurabh.ray@nyu.edu};
\href{https://orcid.org/0009-0005-6708-125X}{ORCID 0009-0005-6708-125X}}}
\affil[1]{National Institute of Science Education and Research, Bhubaneswar, India}
\affil[2]{Homi Bhabha National Institute, Training School Complex, Anushakti Nagar, Mumbai 400094, India}
\affil[3]{IIIT-Delhi, India}
\affil[4]{NYU Abu Dhabi, United Arab Emirates}
\date{}

\graphicspath{{./fig/}}

\newcommand{\R}{\mathcal{R}}
\newcommand{\B}{\mathcal{B}}

\newcommand{\opt}{\textsc{Opt}}
\newcommand{\local}{\textsc{Local}}

\newcommand{\hide}[1]{}

\newcommand{\Tr}{\operatorname{Tr}}
\newcommand{\exc}{\operatorname{exc}}
\newcommand{\se}{\textsc{se}}
\newcommand{\ap}{\textsc{ap}}
\begin{document}

\maketitle
\begin{abstract}

Packing and covering problems for geometric regions have been studied under
many notions of complexity, including VC-dimension, union complexity,
shallow-cell complexity, and fatness.  Although these restrictions often yield
constant-factor approximation algorithms, they do not by themselves generally
lead to PTASs.  A recurring feature of known hardness constructions is that
one region may be \emph{pierced} by many others: a region $B$ pierces $A$ when
$A\setminus B$ is disconnected.

We study geometric instances through the \emph{piercing degree}.  Since
piercing is symmetric for Jordan regions, this is the maximum degree of the
corresponding piercing graph.  Our main result is that, for
every fixed piercing degree, the standard local-search algorithms give PTASs
for the unweighted \emph{Discrete Independent Set} and \emph{Set Cover}
problems.  The proof constructs a sublinear balanced separator for an
appropriate locality graph and applies it adaptively throughout the recursive
local-search analysis.  This guarantee depends only on the piercing degree; in
particular, it places no bound on the number of components created by an
individual piercing pair.

We also prove a polynomial shallow-trace bound depending only on the piercing
degree.  As consequences, for every fixed piercing degree, weighted Set Cover
admits a deterministic $C_r$-approximation and weighted Discrete Independent
Set admits a deterministic $O(r+1)$-approximation.  These
results extend the known guarantees for non-piercing families and apply, for
example, to axis-parallel rectangles when every rectangle is pierced by only a
bounded number of other rectangles.

\end{abstract}


\section{Introduction}
\label{sec:intro}

Given a finite set $X$ and a collection $\mathcal{S}$ of subsets of
$X$ whose union covers $X$, the \emph{Set Cover} problem asks for a smallest
subcollection $\mathcal{S}'$ that also covers $X$.  In the \emph{Discrete
Independent Set} problem, the objective is to select a largest subcollection
$\mathcal{S}'\subseteq\mathcal{S}$ such that each element of $X$ is contained
in at most one set in $\mathcal{S}'$.  Both problems have been extensively
studied, and tight bounds on their approximability are known in the general
setting~\cite{feige1998threshold,lov1975notdef,hastad1999clique}.

Improved approximation algorithms can often be obtained when these problems
are restricted to geometric settings.  Such improvements are typically
enabled by bounding an appropriate notion of \emph{complexity} of the
underlying geometric objects.  Well-studied complexity measures include
dimension, VC-dimension, union complexity and its generalization shallow-cell
complexity, and fatness, among others~\cite{mustafa2022sampling,
chan2012weighted}.  In particular, small
VC-dimension or shallow-cell complexity implies the existence of small
$\epsilon$-nets~\cite{mustafa2022sampling}, which in turn yield approximation
algorithms via LP-rounding.

A representative example is the Set Cover problem defined by unit disks and
points in the plane, which admits a PTAS~\cite{BasuRoy2018}.  In contrast, when
disks are replaced by regions that remain geometrically close to disks, such as
nearly congruent, nearly circular ellipses, the problem becomes APX-hard despite remaining small
with respect to the above complexity measures~\cite{DBLP:journals/comgeo/ChanG14}.
Chan and Grant~\cite{DBLP:journals/comgeo/ChanG14} gave several further examples
of simple geometric objects, including axis-parallel strips and wedges, for
which both Set Cover and Discrete Independent Set are APX-hard; see
also~\cite{Har-PeledQ15}.  A key structural feature of these constructions is
that many pairs of regions \emph{pierce} one another.

In contrast, PTASs are known for a wide range of geometric packing and
covering problems.  In many cases, these algorithms use a simple local-search
framework~\cite{ChanH12,mustafa2010improved,DBLP:conf/walcom/AschnerKMY13,
gibson-pirwani,BasuRoy2018,RR18}.  A common condition in these settings is that
the regions are \emph{non-piercing}\footnote{Examples include disks,
axis-parallel squares, axis-parallel unit-height rectangles, and homothets of a
convex region.}, meaning that the difference of any region with another is
connected; see Definition~\ref{defn:nonpiercing}.  This contrast motivates the
central question of this paper: can one quantify the amount of piercing and
thereby interpolate between the non-piercing regime and unrestricted families?

\smallskip\noindent
\emph{Quantifying piercing.}
For two regions $A$ and $B$, we say that $B$ pierces $A$ if $A\setminus B$ is
disconnected.  Piercing is symmetric for Jordan regions.  The
\emph{piercing degree} $r$ of a family is the maximum, over all its regions
$A$, of the number of regions that pierce $A$.  Our results require no bound on the
number of connected components of $A\setminus B$: one piercing pair may have
arbitrarily complicated boundary interaction.

Our main result, stated informally below, shows that the piercing degree alone
is enough to recover local-search PTASs.  The exchange radius has a polynomial
dependence on $1/\epsilon$ for every fixed $r$.

\begin{theorem*}[Informal]
Let $P$ be a finite point set and let $\Gamma$ be a set of $n$ regions with
piercing degree at most $r$.  For
every $\epsilon\in(0,1]$, $t$-local search with
$t=O(((r+1)/\epsilon)^4\log^3(2+(r+1)/\epsilon))$ yields a
$(1+\epsilon)$-approximation for Set Cover and a
$(1-\epsilon)$-approximation for Discrete Independent Set.  The running time is
$|P|n^{O(t)}$.
\end{theorem*}

Thus the algorithms run in polynomial time for constant $r$ and give QPTASs
when $r$ is polylogarithmic in $n$.

\smallskip\noindent
\emph{Applications.}
Many geometric optimization problems employ idealized shape models: sensor
coverage by disks or wedges, map labels by rectangles, and service regions
defined by $\ell_p$-norm balls.  Non-piercing regions~\cite{RR18}
allow irregular and non-convex shapes, but the non-piercing requirement can be
too restrictive.  The piercing degree gives a topological relaxation in which
each region is pierced by only a bounded number of other regions, while each
such interaction may be arbitrarily complicated.  We use these
examples as motivation rather than as a claim that the piercing degree is
always small in applications.

This condition is incomparable with low density~\cite{Har-PeledQ17}: it is
topological and allows arbitrarily deep nesting, while a low-density family
may contain one sufficiently long region pierced by many well-separated
smaller regions.

\smallskip\noindent
\emph{Axis-parallel rectangles.}
Both Set Cover and Discrete Independent Set are APX-hard for axis-parallel
rectangles~\cite{DBLP:journals/comgeo/ChanG14}.  For continuous Independent Set,
polynomial-time constant-factor
approximations~\cite{DBLP:conf/focs/Mitchell21,DBLP:conf/soda/GalvezKMMPW22} and a
QPTAS~\cite{AW14} are known.  For Set Cover by rectangles and points,
$\epsilon$-net methods give an approximation factor logarithmic in the optimum
cover size~\cite{BronnimannG95,EvenRS05}.  Our results give PTASs for both unweighted
problems when each rectangle is pierced by only a bounded number of the other
rectangles.

\smallskip\noindent
\emph{Weighted problems and shallow traces.}
The same parameter also controls the number of shallow traces.  For every
subfamily of $m$ regions with piercing degree at most $r$, the number of
distinct traces containing at most $k$ regions is
$O_r(1+m(k+1)^{4r+3})$, where $O_r(\cdot)$ hides constants that depend only on
$r$.  This bound, together with quasi-uniform
sampling~\cite{chan2012weighted}, yields a deterministic $C_r$-approximation for
weighted Set Cover, where $C_r$ is a constant for every fixed $r$.  A direct
LP-rounding argument yields a deterministic $O(r+1)$-approximation for
weighted Discrete Independent Set.  We do not claim a uniform
polynomial dependence of $C_r$ on $r$.

\smallskip\noindent
\emph{Technical contribution.}
Our unweighted algorithms retain the standard exhaustive local-search
procedure; the main work is in its analysis.  For any two feasible solutions,
we consider a minimum locality graph whose edges certify the point constraints
shared by the two solutions.  Bounded piercing allows its vertices to be
partitioned into a bounded number of non-piercing families, to which the planar
support theorem of Raman and Ray~\cite{RR18} applies.  This gives hereditary
linear edge density and bounded VC-dimension.

Linear density alone does not imply the separators needed for local search.
We draw the locality edges inside the regions and control the parity of
crossings.  The odd-crossing bisection theorem of Pach and
T{\'o}th~\cite{PachToth} then gives a sublinear balanced separator after large
traces are removed by an $\epsilon$-net.  Recomputing this separator for every
restricted locality instance gives an adaptive recursive division with
sublinear total overlap, which is exactly what the local-exchange argument
requires.  The shallow-trace bound used for the weighted results follows from
linear sparsity of two-element traces and the signature theorem of Ackerman,
Keszegh, and P\'alv\"olgyi~\cite{AKP}.

\begin{remark}[Correction to the conference version]
The conference version~\cite{banik_et_al:LIPIcs.ICALP.2026.21} used a
lens-bypassing argument whose key monotonicity lemma is false; hence its
normalization and the resulting arrangement decomposition were not
established. Its shallow-trace proof also failed to account for collisions
between distinct traces after sampling. These are proof errors, not
counterexamples to the two unweighted PTAS statements. The present version
avoids lens bypassing: it obtains the required decomposition from locality
graphs through an odd-crossing separator and adaptive recursion, and proves
the shallow-trace bound using injective signatures. The statements and bounds
in this version are the corrected ones.
\end{remark}

\subsection{Related Work}
\label{sec:related}
For packing problems with \emph{fat objects}\footnote{There are many definitions of fatness, but a definition that is sufficient for us is that a region is fat if the
ratio of the radii of its circumscribing and inscribing balls is bounded by a constant.} in bounded dimension, e.g., balls or cubes, Erlebach et al.~\cite{ErlebachL10} obtained a PTAS building on the work of Hochbaum and Maass~\cite{hochbaum-mass} for unit disks.
For regions that are not fat, Adamaszek and Wiese~\cite{AW14} obtained a QPTAS (including weighted settings)
for regions in the plane under the restriction that each region in the set is path-connected.
In the discrete setting, Chan and Grant~\cite{DBLP:journals/comgeo/ChanG14} obtained hardness results for packing
and covering problems with simple geometric regions in the plane such as axis-parallel strips.

For covering problems, approximation algorithms primarily build on $\epsilon$-nets. Haussler and Welzl~\cite{HausslerW87} showed $\epsilon$-nets of size $O(d/\epsilon\log(1/\epsilon))$ for set systems with bounded VC-dimension, leading to logarithmic-factor approximations by Br{\"o}nnimann and Goodrich~\cite{BronnimannG95} and Even et al.~\cite{EvenRS05}. Clarkson and Varadarajan~\cite{ClarksonV07} improved $\epsilon$-nets for geometric systems, culminating in Varadarajan's~\cite{V10} \emph{quasi-uniform sampling}, which
was refined by Chan et al.~\cite{chan2012weighted} for systems with shallow-cell complexity $\phi(n)$. Their result yields corresponding small $\epsilon$-nets. It also gives an approximation factor $O(\log\phi(n))$ for weighted Set Cover, though with large constants even in simple cases.

For \emph{unweighted} problems, the unifying tool has been the \emph{local search framework}, which iteratively improves a feasible solution with small changes. This yielded PTASs for independent set~\cite{ChanH12}, hitting set~\cite{mustafa2010improved}, and Set Cover and dominating set~\cite{BasuRoy2018}. Raman and Ray~\cite{RR18} showed that set systems from \emph{non-piercing regions} admit planar supports, implying PTAS results for many packing and covering problems via local search, and also for demand/capacitated versions~\cite{DBLP:journals/dcg/RamanR22}. Raman and Singh~\cite{raman2025supportsouterplanarboundedtreewidth} extended these results to higher-genus surfaces. The limitations of the local search approach were studied by
Jartoux and Mustafa~\cite{Mustafa18}.

\section{Preliminaries}
\label{sec:prelims}
We use the term ``region'' for the closed Jordan domain bounded by a simple
Jordan curve. Thus every region is simply connected. We denote the boundary
of a region $\alpha$ by $\partial\alpha$.
No bound is imposed on the number of intersections of two boundaries or on
the number of components of $\alpha\setminus\beta$ for any two regions
$\alpha$ and $\beta$.

For a graph $H=(V,E)$, we write $E(H)=E$, $H[W]$ for the subgraph induced by
$W\subseteq V$, and $H-W=H[V\setminus W]$. For $A\subseteq V$, its open
neighborhood is
$N_H(A)=\{v\in V\setminus A:uv\in E(H)\text{ for some }u\in A\}$.
\begin{definition}[Piercing and non-piercing~\cite{RR18}]
\label{defn:nonpiercing}
We say that $\beta$ \emph{pierces} $\alpha$ if
$\alpha\setminus\beta$ is disconnected.
A set $\mathcal{F}$ of regions is \emph{non-piercing} if
$\alpha\setminus\beta$ is connected for every ordered pair
$\alpha,\beta\in\mathcal{F}$.
\end{definition}

\begin{lemma}[Symmetry of piercing]
\label{lem:piercing-symmetry}
For Jordan regions, $\alpha$ pierces $\beta$ if and only if $\beta$ pierces
$\alpha$.
\end{lemma}

\begin{proof}
Suppose first that $\partial\alpha$ and $\partial\beta$ intersect in $2h>0$
points. The $h$ arcs of $\partial\beta$ in the interior of $\alpha$ are
pairwise disjoint crosscuts of the Jordan disk $\alpha$, and hence partition
it into $h+1$ cells. If $a$ of these cells are contained in $\beta$, the
remaining $h+1-a$ cells are the components of
$\alpha\setminus\beta$. The $a$ cells contained in $\beta$ are precisely the
components of the common interior. Reversing the roles of $\alpha$ and
$\beta$ therefore gives the same numbers $h$ and $a$, so
$\beta\setminus\alpha$ also has $h+1-a$ components. Thus $\alpha\setminus\beta$ and
$\beta\setminus\alpha$ have the same number of components. If the boundaries do not intersect, the regions
are disjoint or one contains the other, and neither difference is
disconnected.
\end{proof}

The simple \emph{piercing graph} $\Pi(\Gamma)$ has one vertex for every
region in $\Gamma$ and an edge $\alpha\beta$ if the two regions pierce.

\begin{definition}[Piercing degree]
\label{defn:rpiercing}
The piercing degree of $\Gamma$ is
\[
\max_{\alpha\in\Gamma}
|\{\beta\in\Gamma\setminus\{\alpha\}:\beta\text{ pierces }\alpha\}|.
\]
We say that $\Gamma$ is an $r$-piercing family if this degree is at most $r$.
\end{definition}
Equivalently, by Lemma~\ref{lem:piercing-symmetry}, the piercing degree of
$\Gamma$ is the maximum degree of $\Pi(\Gamma)$.

Let $P$ be the finite set of input points. For $p\in P$, let
$\Gamma(p)=\{\alpha\in\Gamma:p\in\alpha\}$ be its \emph{trace}. For
$U\subseteq\Gamma$, write $\Tr_U(p)=\Gamma(p)\cap U$. In the dual set system
defined by $\Gamma$ and $P$, points with identical traces define the same
hyperedge, and we identify them.
We assume that no point of $P$ lies on a region boundary.

For the analysis of the algorithms we present later, we need to consider families with red and blue colour
classes $\R$ and $\B$, so
$U=(U\cap\R)\mathbin{\dot\cup}(U\cap\B)$. A trace is \emph{mixed} if it
contains a region of each colour.

\begin{definition}[Locality graph]
\label{def:locality}
A bipartite graph $L_U$ with vertex set $U$ and
$E(L_U)\subseteq(U\cap\R)\times(U\cap\B)$ is a
\emph{locality graph} for a specified collection of traces restricted to $U$
if every mixed trace contains an edge of $L_U$.
Among all locality graphs for the specified traces, one with the smallest
number of edges is called a \emph{minimum locality graph}.
\end{definition}

\begin{definition}[Shallow-cell complexity~\cite{chan2012weighted}]
The dual set system has shallow-cell complexity $\phi$ if, for every
$U\subseteq\Gamma$ with $|U|=m$ and every $k$, the number of distinct traces
$\Tr_U(p)$ of size at most $k$ is $O(m\phi(m)\poly(k))$.
\end{definition}

\noindent
In particular, if $\phi$ is a constant function, we say that the shallow-cell complexity is constant. 

For a graph $H$, a vertex separator $X$ is \emph{$2/3$-balanced} if every
component of $H-X$ has at most $2|V(H)|/3$ vertices. An edge bisection of a
graph is \emph{$2/3$-balanced} if it partitions the
vertex set into two parts, each containing at most two thirds of the
vertices; its size is the number of edges between the parts. The
odd-crossing bisection theorem of Pach and T{\'o}th~\cite{PachToth} states that a
graph $H$ drawn in the plane has a $2/3$-balanced edge bisection of size
\[
O\left(\log(2+|V(H)|)
\sqrt{\operatorname{odd-cr}(H)+\sum_{v\in V(H)}\deg_H(v)^2}\right).
\]
Here $\operatorname{odd-cr}(H)$ is the minimum, over all plane drawings of
$H$, of the number of unordered pairs of edges that cross an odd number of
times. If $q(D)$ is the number of independent such pairs in a particular
drawing $D$, then
\[
\operatorname{odd-cr}(H)\le q(D)+
\sum_{v\in V(H)}\binom{\deg_H(v)}2
\le q(D)+\frac12\sum_{v\in V(H)}\deg_H(v)^2.
\]

\section{Main Results}
\label{sec:contribution}
In the rest of the paper, $\Gamma$ is assumed to be an $r$-piercing family
and $n=|\Gamma|$. We assume that $\Gamma$ covers $P$ when considering Set
Cover. In the
\emph{Discrete Independent Set} problem, the objective is to select a largest
subset of $\Gamma$ such that each point of $P$ is contained in at most one
selected region. In the \emph{Set Cover} problem, the objective is to select a
smallest subset of regions covering every point of $P$.

The $t$-local search algorithm used in the following theorems is described in
Section~\ref{sec:alg}.

\begin{remark}[Combinatorial implementation]
The local-search algorithms use only the incidence relation between $P$ and
$\Gamma$; they do not require a geometric representation of the points or
regions. The geometry is used only in the analysis, through locality graphs
and their separators.
\end{remark}

\begin{restatable}{thm}{DIS}[Discrete Independent Set]
\label{ptas-independent-set}
For every $\epsilon\in(0,1]$, the $t$-local search algorithm, with $t$ a
sufficiently large constant multiple of
$((r+1)/\epsilon)^4\log^3(2+(r+1)/\epsilon)$, yields a
$(1-\epsilon)$-approximation for Discrete Independent Set. The algorithm runs
in time $|P|n^{O(t)}$.
\end{restatable}

\begin{restatable}{thm}{SC}[Set Cover]
\label{ptas-set-cover}
For every $\epsilon\in(0,1]$, the $t$-local search algorithm, with $t$ a
sufficiently large constant multiple of
$((r+1)/\epsilon)^4\log^3(2+(r+1)/\epsilon)$, yields a
$(1+\epsilon)$-approximation for Set Cover. The algorithm runs in time
$|P|n^{O(t)}$.
\end{restatable}

\begin{restatable}{thm}{SCC}[Shallow traces]
\label{thm:scc}
Let $S\subseteq\Gamma$, let $k\ge1$, and let $\nu_{\le k}(S)$ be the number
of distinct traces $\Tr_S(p)$ of size at most $k$. Then
\[
\nu_{\le k}(S)\le
2+(r+2)^{O(r+1)}|S|(k+1)^{4r+3},
\]
where the constant hidden in the exponent is absolute.
Consequently, for every fixed $r$, the dual set system has constant
shallow-cell complexity.
\end{restatable}

\begin{restatable}{thm}{WSC}[Weighted problems]
\label{thm:WSC}
Let $w:\Gamma\to\mathbb{N}$ be a weight function. For every fixed $r$:
\begin{enumerate}[(i)]
\item weighted Set Cover admits a deterministic polynomial-time
      $C_r$-approximation, where $C_r$ depends only on $r$; and
\item weighted Discrete Independent Set admits a deterministic polynomial-time
      $O(r+1)$-approximation.
\end{enumerate}
\end{restatable}

\begin{remark}[Derandomization]
The quasi-uniform-sampling algorithm used for weighted Set Cover is
derandomized by Chan et al.~\cite{chan2012weighted}. For weighted Discrete
Independent Set, the rounding argument below can be derandomized by the
method of conditional expectations.
\end{remark}

\section{Locality Separators}
\label{sec:locality}

\subsection{Sparse locality graphs and bounded VC-dimension}

We use the planar-support theorem of Raman and Ray~\cite{RR18}: the dual set
system of a non-piercing family has a simple planar support in which
every nonempty trace induces a connected subgraph. In particular, every
two-element trace is an edge of the support.

For $S\subseteq\Gamma$, let $D(S)$ be the graph on $S$ in which
$\alpha\beta$ is an edge iff $\Tr_S(p)=\{\alpha,\beta\}$ for some $p\in P$.

\begin{lemma}[Hereditary density]
\label{lem:density}
The following statements hold with the absolute constant $c=25$.
\begin{enumerate}[(i)]
\item Every $S\subseteq\Gamma$ satisfies
      $|E(D(S))|\le c(r+1)|S|$.
\item Let $F$ be a minimum locality graph for any subcollection of traces on
      a two-coloured family $V\subseteq\Gamma$. Every $W\subseteq V$
      satisfies $|E(F[W])|\le c(r+1)|W|$.
\end{enumerate}
\end{lemma}

\begin{proof}
Fix an arbitrary order on a vertex set $W$. Process the vertices in this order,
and independently consider each vertex with probability
$\theta=1/(2(r+1))$; retain it if it has no previously retained piercing
neighbour. The retained set $J$ is pairwise non-piercing. This is the
standard random-thinning argument of Clarkson and Shor~\cite{clarkson1989applications}.
A non-piercing pair in $W$ is retained with probability at least
$\theta^2(1-\theta)^{2r}\ge \theta^2/4$, while
$\mathbb{E}|J|\le \theta|W|$.

For (i), a non-piercing edge of $D(W)$ whose endpoints survive remains a
two-element trace in $D(J)$. Since $D(J)$ is contained in a planar support,
it has at most $3|J|$ edges. Comparing expectations bounds the number of
non-piercing edges of $D(W)$ by $12|W|/\theta=24(r+1)|W|$. The piercing
edges number at most $r|W|/2$, proving (i) with $c=25$.

For (ii), minimum cardinality gives the following replacement property: the
edges of $F[W]$ may be replaced by any locality graph for the traces restricted
to $W$. Indeed, an original trace either keeps an $F$-edge with an endpoint
outside $W$, or all its certifying edges lie in $W$; in the latter case its
restriction to $W$ is mixed and the replacement certifies it. If $J$ is the
retained set above, the planar support on $J$, with only its bichromatic edges
kept, is a locality graph and has at most $3|J|$ edges. Thus
$|E(F[J])|\le3|J|$. The same expectation calculation bounds the non-piercing
edges of $F[W]$, and the piercing edges again number at most $r|W|/2$.
\end{proof}

\begin{lemma}[VC-dimension]
\label{lem:vc}
The dual set system, and every restriction of it, has VC-dimension at most
$4(r+1)$.
\end{lemma}

\begin{proof}
Suppose that $Q\subseteq\Gamma$ is shattered. The piercing graph on $Q$ has
maximum degree at most $r$, and hence has a proper colouring with at most
$r+1$ colours. Let $J$ be a largest colour class. The regions in $J$ are
pairwise non-piercing, $|J|\ge |Q|/(r+1)$, and $J$ remains shattered. Hence
every pair of vertices of $J$ is a two-element trace. Every planar support
must therefore contain $K_{|J|}$, so $|J|\le4$ and
$|Q|\le4(r+1)$.
\end{proof}

\subsection{A separator from essential traces}

Let $F$ be a minimum locality graph on $U$ for traces restricted to $U$.
Every edge $e\in E(F)$ has an
\emph{essential witness} $p_e$: if $S_e=\Tr_U(p_e)$, then $S_e$ contains the
endpoints of $e$ and $E(F[S_e])=\{e\}$. Otherwise $e$ could be deleted.
Witnesses for different edges are necessarily distinct: if $p_e=p_f$, then
$S_e=S_f$, and $E(F[S_e])=\{e\}=\{f\}$, so $e=f$. Let
\[
T(F)=\sum_{e\in E(F)}|S_e|.
\]

We construct a drawing of the graph where edges may cross multiple times.
The following standard smoothing operation removes self-intersections while
preserving the relevant parities; see also
Pelsmajer, Schaefer, and \v{S}tefankovi\v{c}~\cite{PelsmajerSchaeferStefankovic2007}.

\begin{lemma}[Parity-preserving smoothing]
\label{lem:smoothing}
Suppose a finite graph is drawn by generic immersed intervals: every edge has
finitely many transverse self-intersections and pairwise intersections, no
edge passes through a nonincident vertex, and there are no triple
intersections. The edge curves can be changed into Jordan arcs without
changing the crossing parity of any two distinct edges.
\end{lemma}

\begin{proof}
Let $z$ be a self-crossing of one edge, visited at parameters $a<b$. Choose a
disk about $z$ meeting only the two local branches and no other edge, vertex,
or self-crossing. Of the two noncrossing smoothings, choose the connected one:
join the prefix branch to the branch entering at the second visit, and join
the branch leaving at the first visit to the suffix branch. The resulting
endpoint-to-endpoint interval traverses the middle subcurve between the visits
in reverse. Every portion outside the disk is retained exactly once. The
chosen self-crossing disappears, no new one is created, and all intersections
with other edges remain unchanged. Iterating proves the lemma.
\end{proof}

\begin{lemma}[Essential-witness separator]
\label{lem:essential-separator}
Let $F$ be a minimum locality graph on $U$, where $|U|=s$, and choose one essential
trace per edge, with total incidence $T=T(F)$. Then $F$ has a
$2/3$-balanced vertex separator of size
$O\bigl(((r+1)s)^{1/3}((r+1)^2s+T)^{1/3}
\log^{2/3}(2+s)\bigr)$.
In particular, if every essential trace has size at most $k\ge1$, the separator has size
$O((r+1)s^{2/3}k^{1/3}\log^{2/3}(2+s))$.
\end{lemma}

\begin{proof}
Let us suppose initially that the maximum degree of a vertex in $F$ is bounded
above by $\Delta_0$, which we will drop later. Let $m=|E(F)|$.
For each non-isolated region $\alpha$, choose a \emph{special edge}
$\se_\alpha$ incident on $\alpha$ minimizing $|S_{\se_\alpha}|$.
Let $p_{\se_\alpha}$ be its corresponding point.

We place an \emph{anchor}  $\ap_\alpha$ representing the vertex corresponding
to a region $\alpha$ in the drawing arbitrarily close to
$p_{\se_\alpha}$ so that $\ap_\alpha$ is contained in exactly the regions in
$S_{\se_\alpha}$ and does not lie on any of the region boundaries.
We also ensure that the anchors for distinct regions are pairwise distinct. This is possible because $U$ is finite and the boundaries of its regions are Jordan curves.

We obtain an initial drawing of $F$ as follows. We  draw 
each edge $e=\alpha\beta$ using two {\em arms}: an arc joining $\ap_\alpha$
to $p_e$ in the interior of $\alpha$, denoted by $\arm_\alpha(e)$, and an arc
joining $\ap_\beta$ to $p_e$ in the interior of $\beta$, denoted by
$\arm_\beta(e)$.
We ensure that the interiors of the arms with a common endpoint $p$ are disjoint in a small neighbourhood of $p$. At this point the two arms of an edge may intersect and make the edge self-intersecting. This is fixed by applying Lemma~\ref{lem:smoothing}. However, in the arguments below, we use the initial drawing,
{\em before} applying Lemma~\ref{lem:smoothing}.

Consider two edges $e$ and $f$ that are independent, i.e., they do not share endpoints. The parity of the pair $\{e,f\}$ is the number of crossings modulo 2 of the curves representing $e$ and $f$ in the drawing.

For any pair of independent edges $e=\alpha\beta$ and $f=\gamma\zeta$ which intersect an odd number of times, one of the arms of $e$ intersects one of the arms of $f$ an odd number of times.
We call such a pair of arms an {\em odd pair}. 
Suppose that $\arm_\alpha(e)$ and $\arm_\gamma(f)$ form an odd pair.
Then, we claim that at least one of the following occurs:
\begin{enumerate}[(a)]
    \item {\em [Witness containment.]} Witness $p_e$ lies in $\gamma$ or witness $p_f$ lies in $\alpha$, or both.
    \item {\em [Anchor containment.]} Anchor $\ap_\gamma$ lies in $\alpha$ or anchor $\ap_\alpha$ lies in $\gamma$, or both.
    \item {\em [Piercing.]} $\gamma$ and $\alpha$ pierce.
\end{enumerate}

To see this, note that if none of the above holds, both $\ap_\alpha$ and $p_e$ lie in the connected region $\alpha\setminus\gamma$ and similarly both $\ap_\gamma$ and $p_f$ lie in the connected region $\gamma\setminus\alpha$. Thus, we can join $\ap_\alpha$ and $p_e$ using an arc $\tau_1$ lying in the interior of $\alpha\setminus\gamma$ creating a closed curve along with the arm of $e$ joining $\ap_\alpha$ and $p_e$.
Similarly, we can join $\ap_\gamma$ and $p_f$ using an arc $\tau_2$ lying in the interior of $\gamma\setminus\alpha$,
creating a closed curve along with the arm of $f$ joining $\ap_\gamma$ and $p_f$. These two closed curves intersect in an even number of points. However, $\tau_1$ and $\tau_2$ do not intersect as they lie in disjoint regions. This implies that the two arms intersect in an even number of points, contradicting our assumption.
\footnote{Related parity observations also appear in Pyrga and
Ray~\cite{PyrgaR08} and Buzaglo, Pinchasi, and
Rote~\cite{BuzagloPinchasiRote2013}.}

Figure~\ref{fig:arm-parity} illustrates the arm construction and the parity
obstruction.

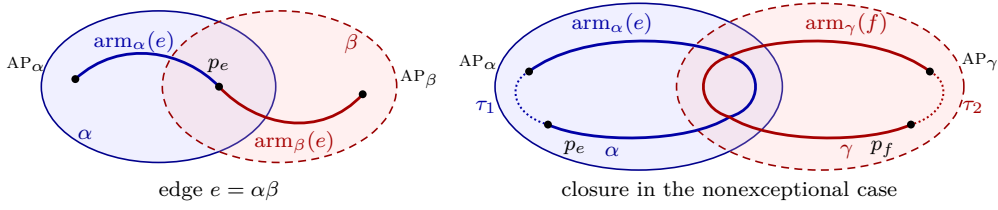
\begin{figure}[t]
\centering
\begin{tikzpicture}[
  x=1cm,y=1cm,
  regA/.style={draw=blue!55!black,fill=blue!25,fill opacity=.22,
               line width=.6pt},
  regB/.style={draw=red!60!black,fill=red!25,fill opacity=.22,
               densely dashed,line width=.6pt},
  armA/.style={draw=blue!65!black,line width=1.05pt},
  armB/.style={draw=red!65!black,line width=1.05pt},
  auxA/.style={draw=blue!65!black,densely dotted,line width=.75pt},
  auxB/.style={draw=red!65!black,densely dotted,line width=.75pt},
  dot/.style={circle,fill=black,inner sep=1.05pt},
  every node/.style={font=\scriptsize}
]
  \draw[regA] (1.75,1.45) ellipse (1.55 and 1.00);
  \draw[regB] (3.35,1.45) ellipse (1.55 and 1.00);
  \coordinate (aa) at (.65,1.55);
  \coordinate (ab) at (4.45,1.35);
  \coordinate (pe) at (2.55,1.45);
  \draw[armA] (aa)
    .. controls (1.25,2.10) and (2.10,1.92) .. (pe);
  \draw[armB] (ab)
    .. controls (3.90,.75) and (3.05,.92) .. (pe);
  \node[dot] at (aa) {};
  \node at (0,1.75) {$\ap_\alpha$};
  \node[dot] at (ab) {};
  \node at (5.20,1.55) {$\ap_\beta$};
  \node[dot,label=above:$p_e$] at (pe) {};
  \node[text=blue!65!black] at (1.43,2.05)
    {$\arm_\alpha(e)$};
  \node[text=red!65!black] at (3.55,.72)
    {$\arm_\beta(e)$};
  \node[text=blue!55!black] at (.78,.82) {$\alpha$};
  \node[text=red!60!black] at (4.30,2.02) {$\beta$};
  \node at (2.55,.08) {edge $e=\alpha\beta$};

  \draw[regA] (8.10,1.45) ellipse (1.90 and 1.10);
  \draw[regB] (10.50,1.45) ellipse (1.90 and 1.10);
  \coordinate (Aa) at (6.65,1.65);
  \coordinate (Pe) at (6.90,.95);
  \coordinate (Ag) at (11.95,1.65);
  \coordinate (Pf) at (11.70,.95);
  \coordinate (xu) at (9.30,1.85);
  \coordinate (xl) at (9.30,1.05);

  \draw[auxA] (Aa)
    .. controls (6.35,1.50) and (6.42,1.08) .. (Pe);
  \draw[auxB] (Ag)
    .. controls (12.25,1.50) and (12.18,1.08) .. (Pf);
  \draw[armA] (Aa)
    .. controls (7.35,2.18) and (8.75,2.12) .. (xu)
    .. controls (9.76,1.72) and (9.76,1.18) .. (xl)
    .. controls (8.75,.70) and (7.45,.70) .. (Pe);
  \draw[armB] (Ag)
    .. controls (11.25,2.18) and (9.85,2.12) .. (xu)
    .. controls (8.84,1.72) and (8.84,1.18) .. (xl)
    .. controls (9.85,.70) and (11.15,.70) .. (Pf);

  \node[dot] at (Aa) {};
  \node at (5.98,1.78) {$\ap_\alpha$};
  \node[dot,label={[xshift=1pt,yshift=-1pt]below right:$p_e$}]
    at (Pe) {};
  \node[dot] at (Ag) {};
  \node at (12.62,1.78) {$\ap_\gamma$};
  \node[dot,label={[xshift=-1pt,yshift=-1pt]below left:$p_f$}]
    at (Pf) {};
  \node[text=blue!65!black] at (6.08,1.18) {$\tau_1$};
  \node[text=red!65!black] at (12.52,1.18) {$\tau_2$};
  \node[text=blue!65!black] at (7.75,2.28)
    {$\arm_\alpha(e)$};
  \node[text=red!65!black] at (10.85,2.28)
    {$\arm_\gamma(f)$};
  \node[text=blue!55!black] at (7.75,.60) {$\alpha$};
  \node[text=red!60!black] at (10.85,.60) {$\gamma$};
  \node at (9.30,.08) {closure in the nonexceptional case};
\end{tikzpicture}
\caption{The arm drawing and its parity obstruction. Left, the edge
$e=\alpha\beta$ is the union of two arms meeting at $p_e$. Right, when
witness containment, anchor containment, and piercing all fail, the dotted
arcs lie in $\alpha\setminus\gamma$ and $\gamma\setminus\alpha$ and close the
two arms. The resulting closed curves have even crossing parity (two
crossings are shown), so the arms cannot form an odd pair.}
\label{fig:arm-parity}
\end{figure}

Next, we bound the number of odd pairs $\{\arm_\alpha(e),
\arm_\gamma(f)\}$. Each such pair has a type $a$, $b$ or $c$ depending on
which of the conditions above holds. If more than one condition holds, we
choose one arbitrarily. We now bound the number of odd pairs of each type.

\begin{enumerate}[(a)]
    \item Suppose that $p_e \in \gamma$. Then, the odd pair is determined by the choice of an edge $e$ and an edge $\gamma\zeta$ where $\gamma$ contains $p_e$. For any choice of $e$, there are $|S_e|$ choices of $\gamma$ and for any such $\gamma$, there are at most $\Delta_0$ choices of $\zeta$. Thus, there are at most $\Delta_0\sum_{e} |S_e| = \Delta_0T$ choices. The case  $p_f \in \alpha$ is also included in this bound.

    \item Suppose that $\ap_\gamma$ lies in $\alpha$. Then, the odd pair is determined by $\ap_\gamma$, the region $\alpha$ containing $\ap_\gamma$, a neighbour $\beta$ of $\alpha$ and a neighbour $\zeta$ of $\gamma$. Once we have chosen $\ap_\gamma$, there are $|S_{\se_\gamma}|$ choices of $\alpha$, $\Delta_0$ choices of $\beta$, and $\deg_F(\gamma)$ choices of $\zeta$. Thus, the total number of choices is
    \[
        \sum_\gamma \Delta_0\deg_F(\gamma)|S_{\se_\gamma}|
        \le \Delta_0\sum_\gamma\sum_{f=\gamma\zeta \in E(F)}|S_f|
        =2\Delta_0T,
    \]
    using the fact that for any edge $f$ incident to $\gamma$, $|S_{\se_\gamma}| \le |S_f|$.
    \item In this case, the odd pair is determined by an edge
    $e=\alpha\beta$, a region $\gamma$ piercing $\alpha$, and an edge
    $f=\gamma\zeta$ incident to $\gamma$.
    The number of such choices is at most $mr\Delta_0$.
    \end{enumerate}

Thus, the number of odd pairs is $O(\Delta_0(rm+T))$, which also bounds the number
of pairs of independent edges crossing an odd number of times.

Apply Lemma~\ref{lem:smoothing} to the drawing. The crossing parity of each
pair is unchanged, so the Pach--T{\'o}th theorem stated in
Section~\ref{sec:prelims} gives a balanced edge bisection of size
$O(\log(s+2)\sqrt{\Delta_0(rm+T)+\Delta_0m})$.
Here $\sum_v\deg_F(v)^2\le2\Delta_0m$. Taking any one endpoint of each edge
in the bisection yields a vertex separator of the same size.

We now drop the degree assumption by adding every vertex of degree greater
than $\Delta_0$ to the separator and applying the preceding argument to the
remaining induced graph. For each retained edge $e$, the restricted witness
$S_e$ intersected with the remaining vertex set still induces only $e$;
these witnesses remain distinct and their total size is at most $T$. Thus the
preceding drawing and parity count apply even though the induced graph need
not itself be a minimum locality graph. By Lemma~\ref{lem:density},
$m=O((r+1)s)$, so there are
$O((r+1)s/\Delta_0)$ such vertices. We therefore obtain a vertex separator
of size
\[
O\left(\frac{(r+1)s}{\Delta_0}+
\log(s+2)\sqrt{\Delta_0((r+1)^2s+T)}\right).
\]
Balancing the two terms gives
\[
\Delta_0=\left(
\frac{(r+1)^2s^2}{\log^2(s+2)((r+1)^2s+T)}
\right)^{1/3}
\]
and the first bound in the statement. If this value is less than $1$, the
trivial separator $U$ already satisfies that bound.

If every essential trace has size at most $k$, then
$T\le k|E(F)|=O((r+1)ks)$. Since
$(r+1)+k\le2(r+1)k$, the first bound simplifies to
$O((r+1)s^{2/3}k^{1/3}\log^{2/3}(2+s))$, proving the second bound.
\end{proof}

\subsection{Locality graph with sublinear separator}
The standard $\epsilon$-net theorem~\cite{HausslerW87} shows the following: a hypergraph of
VC-dimension $d$ on $s$ vertices has a set of
$O(d(s/k)\log(2+s/k))$ vertices meeting every hyperedge of size greater than $k$.
We now use this to handle large traces.

\begin{theorem}[locality graph with sublinear separator]
\label{thm:one-shot}
Let $U$ be a subfamily of a two-coloured $r$-piercing family with $|U| = s$. For
any finite collection of traces restricted to $U$, there is a locality graph
$L_U$ with a $2/3$-balanced vertex separator of size
$O((r+1)s^{3/4}\log^{3/4}(2+s))$. 
\end{theorem}

\begin{proof}
The theorem is immediate for  $s = O(1)$. Otherwise put
$k=\lceil s^{1/4}\log^{1/4}(2+s)\rceil$. By Lemma~\ref{lem:vc} and the
$\epsilon$-net theorem, there is $Z\subseteq U$ of size
$
O\left((r+1)\frac{s}{k}\log\left(2+\frac{s}{k}\right)\right)
$
meeting every restricted trace of size greater than $k$.

Let $W=U\setminus Z$. For the mixed traces not intersecting $Z$, all of which have size at most $k$, choose a minimum locality graph $F$ on $W$.  By Lemma~\ref{lem:density}, $F$ has $O((r+1)s)$ edges. Its essential traces have
total incidence $O((r+1)ks)$, so Lemma~\ref{lem:essential-separator} gives a
balanced separator $Y$ of size
$O((r+1)s^{2/3}k^{1/3}\log^{2/3}(2+s))$.
To construct a locality graph $L_U$ for all traces in $U$
we need to add to $F$ a mixed edge of every mixed trace intersecting $Z$ (which we ignored when constructing $F$).
For each trace intersecting $Z$, we add an arbitrary mixed edge involving an element of $Z$. Then, $Y \cup Z$ is a $2/3$-balanced separator for $L_U$ and has size $|Z| + |Y| = O((r+1)s^{3/4}\log^{3/4}(2+s))$.
\end{proof}

\subsection{Adaptive recursion}
\label{sec:adaptive}
Our goal now is to prove a decomposition result similar to Frederickson's
result for planar graphs~\cite{Frederickson87}, but retaining only the
properties essential for the analysis of local search algorithms.

Let $V\subseteq\Gamma$ be two-coloured, and fix a finite collection
$\mathcal T$ of traces on $V$. Fix a target group size $t$ which is
sufficiently large as a function of $r$. We apply
Theorem~\ref{thm:one-shot} repeatedly and build a recursion tree as follows.
We start with the bag $V$ at the root. As long as a leaf bag $U$ has size
greater than $t$, we apply Theorem~\ref{thm:one-shot} to the restrictions
$T\cap U$, $T\in\mathcal T$, and obtain a locality graph $L_U$ and a balanced
separator $X_U$. For $t$ sufficiently large as a function of $r$,
$|X_U|\le |U|/12$. We group the components of $L_U-X_U$ into two unions $A$
and $B$ so that both $U_1=A\cup X_U$ and $U_2=B\cup X_U$ have between
$|U|/4$ and $3|U|/4$ vertices. Indeed, if some component has size at least
$|U|/4$, take it for $A$; otherwise combine components greedily until their
union first has size at least $|U|/4$. The $2/3$-balance and the bound on
$|X_U|$ give the claimed sizes in either case. We make $U_1$ and $U_2$ the
two children of $U$. When this process stops, let $\mathcal C$ be the
collection of leaf bags and $\mathcal U$ the collection of non-leaf bags.
Every leaf bag has size at most $t$.

For a vertex $v \in V$, let $\mu(v)$ denote the number of leaf bags containing $v$. 
Define
the total excess by
\[
\exc(\mathcal{C})=\sum_{C\in\mathcal{C}}|C|-|V|
=\sum_{v\in V}(\mu(v)-1).
\]
A vertex is \emph{interior} to a leaf bag if that is the only leaf bag containing it. 
All other vertices are {\em boundary} vertices.

\begin{lemma}[Aggregate overlap]
\label{lem:overlap}
There are $O(1+|V|/t)$ leaf bags, each of size at most $t$, whose union is
$V$.
Moreover,
\[
\exc(\mathcal{C})=
\sum_{U \in \mathcal{U}}|X_U|
=O\left((r+1)\frac{\log^{3/4}(2+t)}{t^{1/4}}|V|\right).
\]
In particular, for any $\delta\in(0,1]$, taking $t$ to be a sufficiently
large constant multiple of
$((r+1)/\delta)^4\log^3(2+(r+1)/\delta)$ makes
$\exc(\mathcal C)\le\delta|V|$.
Furthermore, the number of boundary vertices is at most $\exc(\mathcal{C})$, and their
total number of leaf occurrences is at most $2\exc(\mathcal{C})$.
\end{lemma}

\begin{proof}
The fact that $\exc(\mathcal{C})=
\sum_{U \in \mathcal{U}}|X_U|$ follows from the observation that 
at any non-leaf bag $U$, 
a separator vertex appears in both children $U_1$ and $U_2$ and all other vertices belong to exactly one of the bags. 
Thus, the number of leaf bags containing a vertex is one more than the number
of times it appears in a separator at a non-leaf bag. To analyze
$\sum_{U\in\mathcal U}|X_U|$, associate the potential
$\psi(|U|)$ with every bag $U$, where
$\psi(s)=s^{3/4}\log^{3/4}(256+s)$. This function has the following
property: if $s_1,s_2\in[s/4,3s/4]$ and $s_1+s_2\ge s$, then
\[
\psi(s_1)+\psi(s_2)-\psi(s)\ge\eta\psi(s)
\]
for an absolute constant $\eta>0$. Indeed,
$\log(256+s/4)\ge(5/6)\log(256+s)$, and the constraints on $s_1,s_2$ give
\[
s_1^{3/4}+s_2^{3/4}
\ge\bigl(4^{-3/4}+(3/4)^{3/4}\bigr)s^{3/4}.
\]
Consequently, one may take
\[
\eta=(5/6)^{3/4}
\bigl(4^{-3/4}+(3/4)^{3/4}\bigr)-1>0.
\]
Thus, if a non-leaf bag $U$ has children
$U_1$ and $U_2$, Theorem~\ref{thm:one-shot} gives
\[
\psi(|U_1|)+\psi(|U_2|)-\psi(|U|)
=\Omega(|X_U|/(r+1)).
\]
Summing these inequalities over the recursion tree and telescoping gives
\[
\exc(\mathcal C)=O\left((r+1)
\sum_{C\in\mathcal C}\psi(|C|)\right).
\]
If the root is a leaf, the claimed bounds are immediate. Otherwise every leaf
bag $C$ has size in $(t/4,t]$. Hence
$|\mathcal C|=O((|V|+\exc(\mathcal C))/t)$ and
$\psi(|C|)\le\psi(t)$, which imply
\[\exc(\mathcal{C})
= O\left((r+1)\frac{\psi(t)}{t}
  (|V|+\exc(\mathcal{C}))\right).
\]
With large enough $t$ as a function of $r$, we obtain
\[\exc(\mathcal{C}) 
= O\left((r+1)\frac{\psi(t)}{t}|V|\right)
= O\left((r+1)\frac{\log^{3/4}(2+t)}{t^{1/4}}|V|\right).
\]

\noindent
As observed before, the number of leaf bags $|\mathcal{C}| = O((|V| + \exc(\mathcal{C}))/t) = O( 1 + |V|/t)$.

Since each boundary vertex belongs to at least $2$ leaf bags, it contributes at least $1$ to the excess. This implies that the number of boundary vertices is at most $\exc(\mathcal{C})$. 
Also, if a boundary vertex $v$ belongs to $\mu(v)\ge2$ leaf bags, it
contributes $\mu(v)-1\ge\mu(v)/2$ to the excess. This implies that the total
number of leaf occurrences of boundary vertices is at most
$2\exc(\mathcal C)$.
\end{proof}

\begin{theorem}[Locality division]
\label{thm:locality-division}
Let $V$ be a two-coloured subfamily of an $r$-piercing family, let
$\mathcal T$ be a finite collection of traces on $V$, and let
$\delta\in(0,1)$. There is an integer
\[
t=O\left(((r+1)/\delta)^4
\log^3(2+(r+1)/\delta)\right)
\]
and a set $X\subseteq V$ such that $V\setminus X$ is the disjoint union of
groups $G_1,\ldots,G_q$ for some $q\ge0$, each of size at most $t$, with the
following properties for some locality graph $L$ for all mixed traces in
$\mathcal T$:
\begin{enumerate}[(i)]
\item every edge of $L$ has both endpoints in $G_i\cup X$ for some $i$;
\item $|X|\le\delta|V|$,
      $|G_i\cup N_L(G_i)|\le t$ for every $i$, and
      $\sum_{i=1}^q|N_L(G_i)|\le2\delta|V|$.
\end{enumerate}
Moreover, $q=O(1+|V|/t)$. 
\end{theorem}

\begin{proof}
If $V=\varnothing$, take $q=0$, $X=\varnothing$, and $L$ to be the empty
graph. Hence assume that $V$ is nonempty.
Choose $t$ as in Lemma~\ref{lem:overlap} so that
$\exc(\mathcal C)\le\delta|V|$, and perform the adaptive recursion above.
Write the leaf bags as $C_1,\ldots,C_\ell$ and let
\[
X=\{v\in V:\mu(v)\ge2\},
\qquad G_j=C_j\setminus X\quad(1\le j\le\ell).
\]
Omit the empty groups and relabel the remaining pairs as
$(C_1,G_1),\ldots,(C_q,G_q)$. These groups are disjoint and cover
$V\setminus X$. Since $G_i\subseteq C_i$, each group has size at most $t$,
and Lemma~\ref{lem:overlap} gives $q\le\ell=O(1+|V|/t)$. It also gives
\[
|X|\le\exc(\mathcal C)\le\delta|V|,
\qquad
\sum_i|C_i\cap X|\le2\exc(\mathcal C)\le2\delta|V|.
\]
Since $\delta<1$, at least one group remains.

It remains to construct a single locality graph. For each original mixed
trace $T\in\mathcal T$, start at the root. Suppose that $T\cap U$ is mixed at
a non-leaf bag $U$. The graph $L_U$ contains a bichromatic edge of $T\cap U$.
The endpoints of this edge occur together in at least one child of $U$:
after deleting $X_U$, an edge cannot join different components, and every
separator vertex occurs in both children. Continue through such a child.
This reaches a leaf bag $C$ in which $T\cap C$ is mixed. Let $L$ be the graph
on $V$ obtained by choosing one bichromatic pair from $T\cap C$ for every
mixed trace $T$. This is a locality graph. If the group associated with $C$
remains, then $C=C_i$ for some $i$, and the chosen edge lies in
$C_i=G_i\cup(C_i\cap X)\subseteq G_i\cup X$.
If that group was omitted, both endpoints of the edge lie in $X$, so the edge
lies in $G_i\cup X$ for any remaining group. This proves (i).

Finally, let $v\in N_L(G_i)$, and choose $u\in G_i$ such that $uv\in E(L)$.
The leaf bag from which this edge was chosen must be $C_i$, since $u$ occurs
in only one leaf bag. Hence $v\in C_i\setminus G_i=C_i\cap X$, and therefore
\[
N_L(G_i)\subseteq C_i\cap X.
\]
Consequently,
\[
|G_i\cup N_L(G_i)|\le |C_i|\le t,
\qquad
\sum_i|N_L(G_i)|
 \le\sum_i|C_i\cap X|
 \le2\delta|V|.
\]
This proves (ii).
\end{proof}

\begin{remark}[Comparison with Frederickson's division]
Theorem~\ref{thm:locality-division} is analogous to Frederickson's division
for planar graphs~\cite{Frederickson87}: after removing a small exceptional
set, the remaining vertices split into small groups, and every relevant edge
is confined to one group together with the exceptional set. Here the bound
on the exceptional neighbourhood is aggregate rather than per group, which is
exactly what the local-search proofs need. Also, the locality graph and its
separator are recomputed at every recursive step; unlike Frederickson's
construction, we do not recurse on a single fixed graph.
\end{remark}

\section{PTAS for Discrete Independent Set and Set Cover}
\label{sec:alg}

Let $\Gamma$ be an $r$-piercing family of regions and let $P$ be a set of points
in the plane.
In this section, we define the local search algorithms and show that they yield
a PTAS for the Discrete Independent Set and Set Cover problems defined by $\Gamma$ and $P$.

\paragraph*{$t$-Local Search.}
For any given parameter $t$, the algorithm is the following. We always maintain a feasible solution $S$
which in the beginning is chosen arbitrarily.
At any stage, we check if it is possible to replace a subset $T\subseteq S$ by
a subset $T'\subseteq\Gamma\setminus S$ so that $|T|,|T'|\le t$ and
$S'=(S\setminus T)\cup T'$ is a feasible solution with a better objective
value than $S$. If this is possible, we replace $S$ by $S'$ and continue. If
such a {\em local improvement} is not possible, we say that the current
solution is {\em $t$-locally optimal}, in which case we return $S$. A
straightforward implementation has running time $|P|n^{O(t)}$.

\DIS*

\begin{proof}
Let $\opt$ be an optimal solution and $\local$ the solution returned by
$t$-local search. Let $I=\opt\cap\local$, put
$R=\opt\setminus I$ and $B=\local\setminus I$, and let $P'$ be the points
not covered by a region in $I$. Colour $R$ red and $B$ blue. Every exchange
used below inserts only regions of $R$, so it remains feasible when the fixed
family $I$ is restored and may be tested against the local optimality of
$\local$. Feasibility of $\opt$ and $\local$ also implies that no region of
$R\cup B$ contains a point already covered by $I$. Thus every conflict among
$R\cup B$ occurs at a point of $P'$, where the trace contains at most one red
and at most one blue region. If $R=\varnothing$, optimality of $\opt$ and
feasibility of $\local$ imply $B=\varnothing$, and the result is immediate;
assume henceforth that $R\ne\varnothing$.

Apply Theorem~\ref{thm:locality-division} to $V=R\cup B$, using the traces of
points in $P'$ and $\delta=\epsilon/9$. Let $X$, $G_1,\ldots,G_q$, and $L$
be the resulting exceptional set, groups, and locality graph. Put
\[
R_i=R\cap G_i,
\qquad B_i=N_L(R_i)\cap B.
\]
Every blue region that conflicts with a member of $R_i$ belongs to $B_i$:
the corresponding trace is a two-element mixed trace, and its unique possible
locality edge joins these two regions. Thus deleting $B_i$ and inserting
$R_i$ is a feasible exchange. Moreover,
\[
R_i\cup B_i\subseteq G_i\cup N_L(G_i),
\]
so both sides have size at most $t$. Local optimality gives
$|R_i|\le|B_i|$.

The groups partition $V\setminus X$. If $b\in B\setminus X$ belongs to $B_i$,
an edge of $L$ joins $b$ to a vertex $u\in R_i\subseteq G_i$.
Property~(i) places both endpoints in $G_j\cup X$ for some $j$. Since
$u\in G_i$ and the groups are disjoint, $j=i$, and hence $b\in G_i$. Thus a
blue region outside $X$ belongs to at most one $B_i$. Every occurrence of a
blue region of $X$ in $B_i$ is counted by $N_L(G_i)$.
Hence, summing over the groups and using
Theorem~\ref{thm:locality-division},
\[
|R|-|R\cap X|
=\sum_i|R_i|
\le\sum_i|B_i|
\le |B|+\sum_i|N_L(G_i)|
\le |B|+2\delta(|R|+|B|).
\]
Since $|R\cap X|\le|X|\le\delta(|R|+|B|)$, it follows that
$|R|\le|B|+3\delta(|R|+|B|)$. Therefore
\[
\frac{|B|}{|R|}\ge
\frac{1-3\delta}{1+3\delta}\ge1-\epsilon.
\]
The choice of $t$ has the order stated in the theorem.
Adding back the common regions in $I$ preserves the approximation guarantee.
\end{proof}

The proof of the Set Cover theorem uses the same recursive division but a
different exchange.

\SC*

\begin{proof}
Let $\opt$ be an optimal cover and $\local$ the cover returned by $t$-local
search.
Let $I=\opt\cap\local$, $R=\opt\setminus I$, and
$B=\local\setminus I$, and let $P'$ be the points not already covered by $I$.
If $R=\varnothing$, then $I$ already covers every point, and local optimality
forces $B=\varnothing$; the result is immediate. Hence assume
$R\ne\varnothing$.
Every trace of a point in $P'$ contains at least one red and one blue region.
Apply Theorem~\ref{thm:locality-division} to $V=R\cup B$, these traces, and
$\delta=\epsilon/9$. Let $X$, $G_1,\ldots,G_q$, and $L$ be the resulting
objects. For each group, put
\[
B_i=B\cap G_i,
\qquad
R_i=R\cap\bigl(G_i\cup N_L(G_i)\bigr).
\]
We claim that $(B\setminus B_i)\cup R_i$ covers every point in $P'$. If this
were false for a point $p$, every blue region containing $p$ would belong to
$G_i$. The locality graph has a bichromatic edge in the trace of $p$. Its
blue endpoint lies in $G_i$, so its red endpoint lies in
$G_i\cup N_L(G_i)$ and hence belongs to $R_i$, a contradiction.
Thus this is a feasible exchange, and both sides have size at most $t$.
Local optimality gives $|B_i|\le|R_i|$.

Summing over the groups and using Theorem~\ref{thm:locality-division} gives
\[
|B|-|B\cap X|
=\sum_i|B_i|
\le\sum_i|R_i|
\le |R|+\sum_i|N_L(G_i)|
\le |R|+2\delta(|R|+|B|).
\]
Since $|B\cap X|\le|X|\le\delta(|R|+|B|)$, we obtain
$|B|\le|R|+3\delta(|R|+|B|)$, and hence
\[
\frac{|B|}{|R|}\le
\frac{1+3\delta}{1-3\delta}\le1+\epsilon.
\]
Again, the choice of $t$ has the order stated in the theorem.
Adding back the common regions proves the theorem.
\end{proof}

\section{Shallow traces}

In this section we study the shallow-cell complexity of an $r$-piercing family of regions.

A hypergraph is $\lambda$-linear, following Ackerman, Keszegh, and
P\'alv\"olgyi~\cite{AKP}, if every restriction to $m$ vertices has at most
$\lambda m$ two-element hyperedges. By Lemma~\ref{lem:density}, the dual set
system is $25(r+1)$-linear. We use two standard consequences of bounded
VC-dimension and hereditary density of two-element hyperedges from their work.

\begin{lemma}[Small hyperedges and injective signatures]
\label{lem:signatures}
Let $\mathcal H=(V,\mathcal E)$ be a $\lambda$-linear hypergraph of
VC-dimension at most $d$, where $\lambda$ and $d$ are fixed.
\begin{enumerate}[(i)]
\item Every restriction to $W\subseteq V$ has at most $b_{\lambda,d}|W|$ distinct
      nonempty hyperedges of size at most $d$, where
      \[
      b_{\lambda,d}=1+\sum_{h=2}^d
      \frac{(2\mathrm e\lambda d)^{h-1}}{h!}.
      \]
\item Every nonempty hyperedge $E$ can be assigned a signature
      $\sigma(E)\subseteq E$, with $|\sigma(E)|\le d$, so that distinct
      nonempty hyperedges have distinct signatures.
\end{enumerate}
\end{lemma}

\SCC*

\begin{proof}
Put $\lambda=25(r+1)$ and $d=4(r+1)$. By
Lemmas~\ref{lem:density} and~\ref{lem:vc}, every restriction is
$\lambda$-linear and has VC-dimension at most $d$. Assign every
distinct nonempty trace $E$ an injective signature $\sigma(E)$ using
Lemma~\ref{lem:signatures}.
We now use a standard Clarkson--Shor sampling argument~\cite{clarkson1989applications}.
Sample every member of $S$ independently with probability $\rho=1/(k+1)$,
obtaining $Y$. A nonempty trace $E$, with $|E|\le k$, \emph{survives} if
$E\cap Y=\sigma(E)$. Its survival probability is at least
\[
\rho^d(1-\rho)^k\ge \mathrm e^{-1}(k+1)^{-d}.
\]
Distinct surviving traces have distinct traces on $Y$, namely their
signatures. At most one nonempty trace has the empty signature, and the empty
trace, if it occurs, is left unsigned. Every other surviving trace is a
distinct nonempty trace in the restriction to $Y$ and has size at most $d$,
so Lemma~\ref{lem:signatures} bounds their number by $b_{\lambda,d}|Y|$. Taking
expectations and using $\mathbb E|Y|=|S|/(k+1)$ gives
\[
\max\{0,\nu_{\le k}(S)-2\}
\le \mathrm e b_{\lambda,d}|S|(k+1)^{d-1}.
\]
Since $\mathrm e b_{\lambda,d}=(r+2)^{O(r+1)}$, this is the stated bound. The
injectivity of the signatures is important: it prevents distinct original
traces from colliding after sampling.
\end{proof}

\begin{corollary}
\label{cor:small-nets}
For fixed $r$, the dual set system has $\epsilon$-nets of size
$O_r(1/\epsilon)$. The classical VC theorem alone gives the weaker explicit
bound $O((r+1)\epsilon^{-1}\log(1/\epsilon))$.
\end{corollary}

\begin{proof}
The first assertion follows from the shallow-cell quasi-uniform-sampling
theorem~\cite{chan2012weighted}; the second follows from
Lemma~\ref{lem:vc} and the classical $\epsilon$-net theorem.
\end{proof}

\section{Weighted problems}

For weighted Discrete Independent Set, write $w_\alpha=w(\alpha)$ and let $H$
be the conflict graph on $\Gamma$: two regions are adjacent if some point of
$P$ belongs to both.
Consider the packing relaxation
\[
\begin{array}{ll}
\text{maximize} & \displaystyle\sum_{\alpha\in\Gamma}w_\alpha x_\alpha,\\[0.3em]
\text{subject to}
& \displaystyle\sum_{\alpha\in\Gamma(p)}x_\alpha\le1
  \quad\text{for every }p\in P,\\
& 0\le x_\alpha\le1\quad\text{for every }\alpha\in\Gamma.
\end{array}
\]
The argument below follows the sampling-and-rounding framework of Chan and
Har-Peled~\cite{ChanH12}, with Lemma~\ref{lem:density} replacing their
union-complexity bound.
\begin{lemma}[Fractional conflict density]
\label{lem:fractional-density}
For every $U\subseteq\Gamma$, every feasible LP solution satisfies
\[
\sum_{\alpha\beta\in E(H[U])}x_\alpha x_\beta
\le c_1(r+1)\sum_{\alpha\in U}x_\alpha
\]
for an absolute constant $c_1$.
\end{lemma}

\begin{proof}
Sample each $\alpha\in U$ independently with probability $x_\alpha/2$,
obtaining $Y$. For every conflict edge $\alpha\beta$, fix a trace containing
both endpoints. The probability that $\alpha$ and $\beta$ are sampled and no
other member of that trace in $U$ is sampled is at least
$x_\alpha x_\beta/8$. This follows from
$\prod_i(1-y_i)\ge1-\sum_i y_i$, applied with
$y_\gamma=x_\gamma/2$ for the other members $\gamma$ of the trace, and the LP
constraint. On this event,
$\alpha\beta$ is a two-element trace in the restriction to $Y$. By
Lemma~\ref{lem:density}, $D(Y)$ has $O((r+1)|Y|)$ edges. Comparing
expectations proves the claim.
\end{proof}

\WSC*

\begin{proof}
For (i), selecting minimum-weight regions covering every point is the
minimum-weight hitting-set problem in the dual set system. Theorem~\ref{thm:scc}
supplies the hereditary shallow profile
$O_r(1+m(k+1)^{4r+3})$. For fixed $r$, the
derandomized quasi-uniform-sampling algorithm of Chan et
al.~\cite{chan2012weighted} gives a polynomial-time $C_r$-approximation. Here
$C_r$ is a constant for every fixed $r$; we do not claim a uniform polynomial
dependence of $C_r$ on $r$.

For (ii), solve the displayed packing LP and discard vertices with
$x_\alpha=0$. Lemma~\ref{lem:fractional-density}, applied to every induced
subgraph, implies that every nonempty induced subgraph contains a vertex whose
neighbour $x$-mass is $O(r+1)$. Indeed, summing each neighbour mass with
weight $x_\alpha$ counts every term $x_\alpha x_\beta$ twice, so the claim
follows by averaging. Repeatedly remove such a vertex and record the
elimination order.

Process the vertices in reverse order. Independently propose $\alpha$ with
probability $\rho_\alpha=x_\alpha/(c_2(r+1))$, for a sufficiently large absolute
constant $c_2$, and accept it if none of its already processed neighbours was
proposed. When $\alpha$ is processed, its already processed neighbours have
total $x$-mass $O(r+1)$. A union bound shows that, conditional on $\alpha$
being proposed, it is accepted with probability at least $1/2$. Thus
\[
\Pr[\alpha\text{ is accepted}]
=\Omega\left(\frac{x_\alpha}{r+1}\right).
\]
The accepted regions form a feasible independent set. Linearity of
expectation gives expected weight
$\Omega((r+1)^{-1}\sum_\alpha w_\alpha x_\alpha)$.

For completeness, this rounding can be derandomized directly. Let
$Z_\alpha$ be the proposal indicator and let $N_H^+(\alpha)$ be the already
processed neighbours of $\alpha$. The weight of the accepted set is
\[
\Phi(Z)=\sum_\alpha w_\alpha Z_\alpha
\prod_{\beta\in N_H^+(\alpha)}(1-Z_\beta).
\]
Under independent proposal probabilities $\rho_\alpha$, every conditional
expectation of $\Phi$ is computable in polynomial time. Fixing the indicators
one at a time to preserve this conditional expectation produces a
deterministic independent set of at least the expected weight, proving (ii).
\end{proof}

\section{Conclusion}
We have shown that the piercing degree alone controls the locality structures
needed for geometric optimization. For every fixed piercing degree, the
standard local-search algorithms yield PTASs for unweighted Set Cover and
Discrete Independent Set. The same parameter gives polynomial shallow-trace
bounds, a constant-factor approximation for weighted Set Cover, and an
$O(r+1)$-approximation for weighted Discrete Independent Set.

The proof separates the algorithm from its analysis: the algorithm remains
the usual exhaustive local search, while minimum locality graphs, parity
separators, and the adaptive recursive division are used only to certify that
a locally optimal solution is globally near-optimal. 

Several questions remain open. In particular, it would be interesting to
obtain comparable guarantees for Hitting Set and to understand how far the
bounded-piercing hypothesis can be relaxed for axis-parallel rectangles and
other classical geometric families.

\paragraph{Use of generative AI.}
The authors used ChatGPT iteratively in developing the results presented in this paper. 
The ideas in the initial draft were corrected, simplified and substantially rewritten by the authors. Claude was also used to review the entire paper. The authors have checked all the proofs and take full responsibility for the correctness of results presented here.

\bibliography{ref}
\end{document}